\documentclass{article}
\usepackage[utf8]{inputenc}
\usepackage{tikz}
\usepackage{siunitx}
\usepackage{amsthm}
\usepackage{amsfonts}
\usepackage{amssymb}
\usepackage{amsopn}
\usepackage[hidelinks]{hyperref}
\usepackage{enumitem}
\usepackage{algorithm,algpseudocode,caption}
\usepackage{amsmath}
\usepackage{amsfonts}
\usepackage{amssymb}
\usepackage{nicefrac}

\usepackage[style=nature]{biblatex}
\theoremstyle{definition}
\newtheorem{definition}{Definition}
\newtheorem*{definition*}{Definition}
\newtheorem{lemma}[definition]{Lemma}

\usepackage[a4paper, total={5.8in, 10in}]{geometry}
\usepackage{enumitem}
\usepackage{braket}
\usepackage{authblk}

\title{Speeding up quantum circuits simulation using ZX-Calculus}
\author[1]{Tristan Cam}
\author[1]{Simon Martiel}
\affil[1]{Atos Quantum Lab, Les Clayes-sous-Bois, France}

\begin{document}
\maketitle
\begin{abstract}
We present a simple and efficient way to reduce the contraction cost of a tensor network to simulate a quantum circuit. We start by interpreting the circuit as a ZX-diagram. We then use simplification and local complementation rules to sparsify it. 
We find that optimizing graph-like ZX-diagrams improves existing state of the art contraction cost by several order of magnitude. In particular, we demonstrate an average contraction cost 1180 times better for Sycamore circuits of depth 20, and up to 4200 times better at peak performance.
\end{abstract}


\section{Introduction}

Simulating quantum circuit via classical algorithms is believed to be an intrinsically hard task. However, as the field of quantum computing quickly grows, so does the need for a better understanding of the limit between what is classically accessible and what is out of reach even for large HPC infrastructures.
Even though the simulation of small circuits (i.e. less than $30$ qubits) is widely accessible through direct linear algebra simulation \cite{guerreschi2020intel,suzuki2021qulacs}, the simulation of larger quantum circuits requires less naive approaches. For instance, the Matrix Product State (MPS) formalism allows to perform lossy simulation of very large quantum systems in shallow circuit regime \cite{vidal2003efficient,schollwock2011density}. Another approach consists in restricting the set of gates to a non-universal one in order to achieve polynomial-time simulation. This method has been successfully applied to Clifford circuits \cite{aaronson2004improved}, circuits with a bounded number of non-Clifford rotations \cite{Bravyi2019simulationofquantum} and match gates/Gaussian operators \cite{jozsa2008matchgates}. Overall, even though all those methods have been successful in their own niche, 
tensor network-based approaches \cite{markov2008simulating,chen2018classical,huang2020classical,lykov2020tensor,gray2021hyper,dudek2019efficient,liang2021fast,vincent2022jet} are the most promising to push back the ``quantum supremacy'' frontier claimed by Google \cite{arute2019quantum} 
and to gain an insight on the true limits of classical computing.
Unlike other approaches, the tensor network contraction framework provides a universal tool to represent and simulate various quantum computations.


\paragraph{Tensor Network contraction.} In this setting, quantum circuits are interpreted as multi-graphs (\textit{networks}) where nodes contain the matrices of the gates of the circuit (\textit{tensors}) and edges represent ``dot products'' between the matrices of two nodes. Such a representation is called a \textit{tensor network}.
Given a tensor network, one can contract each edge of the graph, iteratively merging nodes, until left with a single node. The order in which we decide to contract a tensor network is critical: different contraction orders can lead to simulation times several orders of magnitude apart. Thus, the order finding step plays an important role in tensor-network-based simulation methods.
Although finding an optimal order of contraction is known to be \textit{NP-hard} \cite{Biamonte_2015}, Markov and Shi defined the contraction complexity of a tensor network \cite{markov2008simulating} and showed that the contraction can be done in a time that scales exponentially only in a certain metric of the graph, based on the \textit{treewidth}. This complexity has later been refined by O'Gorman using a different metric \cite{o2019parameterization}.

\paragraph{ZX-calculus and tensor networks.} The ZX-calculus is a graphical language that can represent discrete quantum processes, and in particular post-selected quantum circuits \cite{Coecke2007graphicalcalculus,vilmart2018nearoptimal}. Terms in this language are  tensor networks that can be rewritten using a set of rules that guarantees two properties:\\
\begin{itemize}
    \item soundness, meaning they preserve the underlying linear operator
    \item completeness, meaning that if two diagrams implement the same operator, one can be rewritten into the other \cite{jeandel2020completeness}.
\end{itemize}
ZX-diagrams form a well-structured family of tensor networks where tensors are sparse either in the computational basis ($Z$-spiders) or in the diagonal basis ($X$-spiders). This structure allows in particular for. tensor decomposition $-$ the action of splitting a tensor into two or more tensors, used to truncate MPS in approximate simulations $-$ ``for free" as it precisely corresponds to a diagrammatic rewriting rule that can be applied anytime.

Thus, it is natural to ask whether it is beneficial to contract ZX-diagrams instead of regular tensor networks. And if so, we want to find out if we can use the rewriting rules in ZX-Calculus to optimize a ZX-diagram before contracting it like a tensor network.

Our contribution is a method to simplify and optimize a quantum circuit using ZX-Calculus before the order finding step. We present a naive heuristic that uses graph theory operations \cite{kotzig1968eulerian,perdrix2006modeles,duncan2020graph} compatible with ZX formalism to rewrite a ZX-diagram and improve a ``proxy" score based on treewidth \cite{markov2008simulating,dumitrescu2018benchmarking,gogate2012complete}.
We show that we can decrease the contraction cost by three orders of magnitude on average $-$ and by a factor of 4200 at peak performance $-$ compared to the state of the art order finding algorithm \cite{huang2020classical} for the largest circuits we studied (depth 20 Sycamore circuits \cite{arute2019quantum}).

In Section \ref{tn sim}, we will recall the links between tensor networks and quantum circuit simulation and provide details on the order-finding heuristic used in this paper. In Section \ref{zx}, we will show how to apply ZX-Calculus in order to change the structure of a target tensor network through local operations.
In Section \ref{strat}, we will explain our strategy to simplify and optimize a tensor network before contraction. Finally we will present the results of our pre-processing in Section \ref{res} before discussing its performances and future works in Section \ref{disc}.

\section{Circuit Simulation and Tensor Networks}
\label{tn sim}


\subsection{Tensor network contraction for circuit simulation}

In a typical (strong) simulation setting, given some quantum circuit $C$ and some computational basis vector $|x\rangle$, the aim is to return the amplitude $\langle x |C|0\rangle$.
 Computing this final amplitude $\langle x | C|0\rangle$ requires performing a large number of generalized matrix-matrix multiplications (called ``tensor dots").

The idea behind tensor network contraction is to generate a multi-graph where nodes are tensors and whose set of edges represent all tensor dot operations that need to be performed. In that setting, the vector $C|0\rangle$ can be efficiently represented by a graph whose structure closely resembles that of $C$. Once this tensor network is built, one can easily add degree one nodes to this network containing $1-$qubit linear forms $\langle 0|$ or $\langle 1|$ in order to obtain a network representing the scalar quantity $\langle x |C|0\rangle$. Figure \ref{fig:tn_example} gives an example of a circuit and its underlying tensor network structure. For simplicity in the rest of this work, we will assume that $x=0^n$.

\begin{figure}[h!]
    \centering
        \begin{tikzpicture}
        \draw (0, 0) node[rectangle, draw, inner sep=2pt](h){$H$};
        \draw (1, 0) node[circle, fill=black, draw, inner sep=1pt](ctrl){};
        \draw (1, -1) node[circle, draw, inner sep=3.5pt](tgt){};
        \draw (1, -1-0.16) -- (1, -1+2*0.124);
        \draw (ctrl) --(tgt);
        \draw (2, -1) node[rectangle, draw, inner sep=2pt](rz){$R_Z(\theta)$};
        
        \draw (3, 0) node[circle, fill=black, draw, inner sep=1pt](ctrl2){};
        \draw (3, -1) node[circle, draw, inner sep=3.5pt](tgt2){};
        \draw (3, -1-0.16) -- (3, -1+2*0.124);
        \draw (ctrl2) --(tgt2);

        \draw (4, 0) node[circle, fill=black, draw, inner sep=1pt](cz1){};
        \draw (4, -1) node[circle, fill=black, draw, inner sep=1pt](cz2){};
        \draw (cz1) --(cz2);
        
        \draw (-1, 0) -- (h) -- (ctrl) -- (ctrl2) -- (cz1) -- (5, 0);
        \draw (-1, -1) -- (rz) -- (cz2) --  (5, -1);
    \end{tikzpicture}~~~~
    \begin{tikzpicture}
        \draw (0, 0) node[circle, draw, inner sep=2pt, fill=black](p0){};
        \draw (0,-1) node[circle, draw, inner sep=2pt, fill=black](p1){};
        \draw (1, 0) node[circle, draw, inner sep=2pt, fill=black](h){};
        \draw(2, -0.5) node[circle, draw, inner sep=2pt, fill=black](cnot1){};
        \draw(4, -0.5) node[circle, draw, inner sep=2pt, fill=black](cnot2){};
        \draw(3, -1) node[circle, draw, inner sep=2pt, fill=black](rz){};
        \draw(5, -0.5) node[circle, draw, inner sep=2pt, fill=black](cz){};
        \draw (6, 0) node[circle, draw, inner sep=2pt, fill=black](lf0){};
        \draw (6, -1) node[circle, draw, inner sep=2pt, fill=black](lf1){};
        \draw (p0) -- (h) --(cnot1) --(cnot2);
        \draw (p1) -- (cnot1) --(rz)--(cnot2);
        \draw (cnot2) ..  controls (4.5, -0.25) .. (cz);
        \draw (cnot2) ..  controls (4.5, -0.75) .. (cz);
        \draw (cz) -- (lf0);
        \draw (cz) -- (lf1);
    \end{tikzpicture}
    \caption{A circuit $C$ and its corresponding tensor network structure, considering a initial product state and final projector.}
    \label{fig:tn_example}
\end{figure}
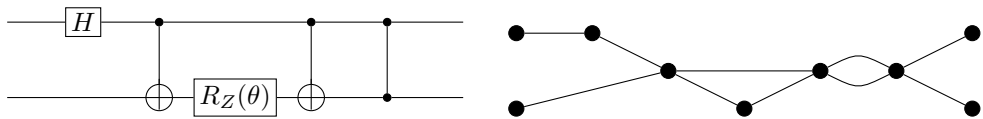

Given such a tensor network structure, one can compute the resulting amplitude by sequentially contracting each edge of the network, each time replacing the (mutli-)edge extremities by a single node. However, the final contraction cost, i.e. the sum of all the FLOPs required to contract each edge in the network, strongly depends on the contraction order of the edges.

In their paper \cite{markov2008simulating}, Markov and Shi show that for any graph $G = (V,E)$, the contraction cost of $G$ is equal to the treewidth of $G^\ast$, the line-graph of $G$. Furthermore, given a tree decomposition of $G^\ast$ of width d, there exists a deterministic algorithm that outputs a contraction ordering $\pi$ with $cc(\pi) \leq d$ in polynomial time (see Algorithm \ref{algo1}).

\subsection{Heuristics for good order extraction}
\label{order}

Our contraction order finder uses a community-based approach described by Gray and Kourtis \cite{gray2021hyper}: find the best partition of our tensor network in communities using the Louvain community detection algorithm \cite{blondel2008fast}, then run QuickBB \cite{gogate2012complete} to find a good tree decomposition of the line-graphs of each individual community, and use Markov and Shi's algorithm to turn the tree decompositions into contraction orders. Due to the relatively small size of communities (80 to 110 nodes for depth 20 Sycamore circuits), the search for an optimal tree decomposition is much easier and the quality of the contraction order is improved. When the actual contraction is performed, each community is collapsed into a single node called \textit{community vertex}. We call this partially contracted network the \textit{metagraph} of the tensor network, according to a certain partition. Thus after adding the individual community contraction orders to the main contraction order, we find a contraction order of the metagraph using the same method. In the end the resulting order is such that all communities will be contracted nearly at no cost (FLOPs-wise),
while most FLOPS come from the contraction of the resulting metagraph.
This is why for the \textit{index slicing} step - a standard method necessary to fit the tensor network contraction in memory by dividing the main contraction task into many sub-tasks \cite{chen2018classical} - we focus on the indices of the intermediate tensors in the metagraph (see Algorithm \ref{algo1}).

\begin{algorithm}
\caption{Community-based order finder}
\label{algo1}
\begin{algorithmic}
\Require{tensor network} 
\Ensure{contraction order}
\Statex
\Function{TreeDecompositionToOrder}{$\mathcal{T}$}\Comment{Markov and Shi's method}
  \State { $\triangleright$ Outputs vertices of $\mathcal{T}$ i.e. vertices of the line-graph: edges in the original graph}
  \While{$\mathcal{T}$ is not the empty graph}                    
    \State {$l$ $\gets$ {leaf of $\mathcal{T}$}}
    \If{$l$ is the last vertex of $\mathcal{T}$}
      \State {\textbf{output} vertices in $B_l$ in any order} \Comment{$B_l$ is the bag of vertices in $l$}
      \State {remove $l$ from $\mathcal{T}$}
    \Else
      \State {$l'\gets$ parent of $l$}
      \If{$B_l\subseteq B_{l'}$}
        \State {remove $l$ from $\mathcal{T}$ and repeat this process}
      \Else
         \State {$e\in B_l-B_{l'}$}
         \State {\textbf{output} $e$ and remove it from $\mathcal{T}$}
      \EndIf
    \EndIf
  \EndWhile
\EndFunction
\Comment{The number of steps is polynomial in the size of $\mathcal{T}$}
\Statex
\Statex
\Function{FindOrder}{$network$}
  \State {$G\gets$ underlying graph from input $network$}
  \State {$P\gets$ \Call{LouvainCommunityPartition}{$G$}}
  \State {$order\gets$ empty list} \Comment{list of edges to be contracted}
  \For{\textbf{each} community $c$ in $P$}
    \State {$lg_c\gets$\Call{LineGraph}{$c$}}
    \State {$\mathcal{T}_c\gets$\Call{QuickBB}{$lg_c$}} \Comment{The tree decomposition heuristic}
    \State {$order$.insert(\Call{TreeDecompositionToOrder}{$\mathcal{T}_c$})}
  \EndFor {\textbf{ each}}
  \State {$mg\gets$\Call{GraphFromCommunityVertices}{$P$}} \Comment{Constructing the metagraph}
  \State {$lg\gets$\Call{LineGraph}{$mg$}}
  \State {$\mathcal{T}\gets$\Call{QuickBB}{$lg$}}
  \State {$order$.insert(\Call{TreeDecompositionToOrder}{$\mathcal{T}$})}
  \State {$slices\gets$\Call{FindSlices}{$mg$, $order$}}\Comment{Greedily slices the indices of the biggest contraction occuring when performing $order$ on $network$}
  \State \Return{$order$, $slices$}
\EndFunction
\end{algorithmic}
\end{algorithm}

\paragraph{Pre-contraction.}
\label{prec}

In most applications, computing a decently good contraction cost can be expensive. 
We would like to make tensor network manipulation easier for computing contraction order/cost, and in particular to estimate the treewidth of its line graph.
Thus we introduce a way to compress our tensor network into a smaller and denser one that encapsulates the structure of the original network while preserving the treewidth of its line-graph.

\begin{definition}[Pre-contraction]
Let $G$ be a tensor network of maximum degree 3. Pre-contracting $G$ consists in merging all \textit{leaves} (degree 1 nodes) with $G$ to their neighbor recursively until $G$ has no \textit{branch} ending in a leaf (only cycles). All cycle edges that do not belong to a triangle in $G$ are then contracted.
\end{definition}
Only a small number of tensor product is required during this pre-contraction step. A ``condensed" version of the original tensor network is obtained, enable easier order finding and further improving contraction cost.

Note that this ``mock contraction" can also be used on a ZX-diagram before computing the \textit{treewidth} of its \textit{line-graph} ($tw(G^\ast)$), yielding an approximation of at least one order of magnitude faster.
In our benchmarks, we use this pre-contraction in conjunction with the QuickBB algorithm \cite{gogate2012complete} to approximate $tw(G^\ast)$ (see Section \ref{appendix} for more details). We call this method $quick\_tw$. 

\section{Tensor Network Simplification using ZX-Calculus}
\label{zx}

In their paper, Duncan et al. \cite{duncan2020graph} provide a brief overview of ZX-calculus:
ZX-calculus is a diagrammatic language that resembles the standard quantum circuit
notation. A ZX-\textit{diagram} consists of:
\begin{itemize}
    \item \textit{wires}: all wires entering the diagram from the left are \textit{inputs} and all wires exiting to the right are \textit{outputs}.
    \item \textit{spiders}: spiders are linear maps which can have any number of input or output wires. There are two varieties: \textit{Z}-spiders depicted as green dots and \textit{X}-spiders depicted as red dots. We will represent Hadamard edges as yellow boxes.
\end{itemize}

Written explicitly in Dirac notation, these linear maps are:
\begin{figure}[h!]
\centering
\includegraphics[scale=0.45]{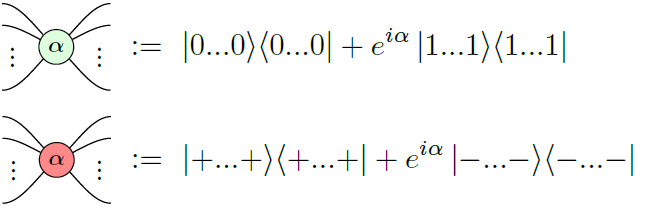}
\end{figure}

Here are a few examples of gates interpreted as ZX-diagrams:

\begin{tikzpicture}[scale=0.75, transform shape]
    \node[shape=circle,draw=black,fill=white!80!green] (A) at (0,0) {$\alpha$};
    \node[shape=circle,draw=white] (r) at (-1.75,0) {\Large $R_{\alpha} =$};
    \coordinate (B) at (-1,0);
    \coordinate (C) at (1,0);
    \draw[-] (A) -- (B);
    \draw[-] (A) -- (C);
    
    \node[shape=circle,draw=black,fill=white!50!red] (A) at (6,-0.35) {};
    \node[shape=circle,draw=black,fill=white!80!green] (D) at (6,0.3) {};
    \node[shape=circle,draw=white] (r) at (3.75,0) {\Large $CNOT =$};
    \coordinate (B) at (5,-0.35);
    \coordinate (C) at (7,-0.35);
    \coordinate (E) at (5,0.3);
    \coordinate (F) at (7,0.3);
    \draw[-] (A) -- (B);
    \draw[-] (A) -- (C);
    \draw[-] (D) -- (E);
    \draw[-] (D) -- (F);
    \draw[-] (A) -- (D);
    
    \node[shape=circle,draw=white] (r) at (9,0) {\Large $H =$};
    \coordinate (A) at (9.7,0);
    \coordinate (B) at (11.2,0);
    \draw[-] (A) -- (B);
    \node[shape=rectangle,draw=black,fill=white!20!yellow] (C) at (10.45,0) {};
    \node[shape=circle,draw=white] (r2) at (11.6,-0.075) {\Large $=$};
    \coordinate (A) at (12,0);
    \coordinate (B) at (13.25,0);
    \draw[-,color=white!30!blue] (A) -- (B);
\end{tikzpicture}

To simulate Google's Sycamore circuits, we need additional 1-qubit gates:

\begin{equation*}
\sqrt{X} =\frac{1}{\sqrt{2}} \begin{bmatrix}
~1 & -i~\\
~-i & 1~
\end{bmatrix}, 
 \sqrt{Y} =\frac{1}{\sqrt{2}} \begin{bmatrix}
~1 & -1~\\
~1 & 1~
\end{bmatrix}, 
 \sqrt{W} =\frac{1}{\sqrt{2}} \begin{bmatrix}
~1 & -\sqrt{i}~ \\
~\sqrt{-i} & 1~
\end{bmatrix}
\end{equation*}

\hspace{-2mm}
\begin{tikzpicture}[scale=0.75, transform shape]
    \node[shape=circle,draw=black,fill=white!50!red,minimum size = 0.55cm,inner sep=0pt] (A) at (1.6,0) {$\frac{\pi}{2}$};
    \node[] (r) at (-3.25,-0.1) {\Large Represented in ZX:};
    \node[] (r) at (-0.25,0) {\Large $\sqrt{X} =$};
    \coordinate (B) at (0.6,0);
    \coordinate (C) at (2.6,0);
    \draw[-] (A) -- (B);
    \draw[-] (A) -- (C);
    
    \node[shape=circle,draw=black,fill=white!80!green,minimum size = 0.55cm,inner sep=0pt] (A) at (6.05,0) {$\frac{\pi}{2}$};
    \node[shape=circle,draw=black,fill=white!50!red,minimum size = 0.55cm,inner sep=0pt] (D) at (5.3,0) {$\frac{\pi}{2}$};
    \node[shape=circle,draw=black,fill=white!50!red,minimum size = 0.55cm,inner sep=0pt] (E) at (6.8,0) {\mbox{-}$\frac{\pi}{2}$};
    \node[] (r) at (3.55,0) {\Large $,~\sqrt{Y} =$};
    \coordinate (B) at (4.55,0);
    \coordinate (C) at (7.55,0);
    \draw[-] (D) -- (B);
    \draw[-] (D) -- (A);
    \draw[-] (E) -- (A);
    \draw[-] (E) -- (C);
    
    \node[] (r) at (8.55,0) {\Large $,~\sqrt{W} =$};
    \node[shape=circle,draw=black,fill=white!50!red,minimum size = 0.55cm,inner sep=0pt] (A) at (11.15,0) {$\frac{\pi}{2}$};
    \node[shape=circle,draw=black,fill=white!80!green,minimum size = 0.55cm,inner sep=0pt] (D) at (10.4,0) {\mbox{-}$\frac{\pi}{4}$};
    \node[shape=circle,draw=black,fill=white!80!green,minimum size = 0.55cm,inner sep=0pt] (E) at (11.9,0) {$\frac{\pi}{4}$};
    \coordinate (B) at (9.65,0);
    \coordinate (C) at (12.65,0);
    \draw[-] (D) -- (B);
    \draw[-] (D) -- (A);
    \draw[-] (E) -- (A);
    \draw[-] (E) -- (C);
\end{tikzpicture}

And the 2-qubit gate parametrized by $\theta$ and $\phi$:

\vspace{3mm}
\begin{tikzpicture}[scale=0.71, transform shape]
    \node[] (r) at (0,0) {
        \Large fSim($\theta,\phi$) = $\begin{bmatrix}
        ~1 & 0 & 0 &0~\\
        ~0 & \cos(\theta) & -i\sin(\theta) &0~\\
        ~0 & -i\sin(\theta) & \cos(\theta) &0~\\
        ~0 & 0 & 0 &e^{-i\phi}~
        \end{bmatrix}$
        \Large =};
    \node[shape=circle,draw=black,fill=white!50!red,minimum size = 0.55cm,inner sep=0pt] (A) at (6.5,0.65) {$\frac{\pi}{2}$};
    \node[shape=circle,draw=black,fill=white!50!red,minimum size = 0.55cm,inner sep=0pt] (AA) at (6.5,-0.65) {$\frac{\pi}{2}$};
    \coordinate (B) at (5.5,0.65);
    \coordinate (C) at (5.5,-0.65);
    \draw[-] (A) -- (B);
    \draw[-] (AA) -- (C);
    \node[shape=circle,draw=black,fill=white!80!green] (D) at (7.25,0.65) {};
    \node[shape=circle,draw=black,fill=white!50!red] (DDD) at (7.25,0) {};
    \node[shape=circle,draw=black,fill=white!80!green,minimum size = 0.5cm,inner sep=0pt] (DDDD) at (7.85,0) {$\theta$};
    \node[shape=circle,draw=black,fill=white!80!green] (DD) at (7.25,-0.65) {};
    \draw[-] (A) -- (D);
    \draw[-] (AA) -- (DD);
    \draw[-] (D) -- (DDD);
    \draw[-] (DD) -- (DDD);
    \draw[-] (DDDD) -- (DDD);
    \node[shape=circle,draw=black,fill=white!50!red,minimum size = 0.55cm,inner sep=0pt] (E) at (8.1,0.65) {\mbox{-}$\frac{\pi}{2}$};
    \node[shape=circle,draw=black,fill=white!50!red,minimum size = 0.55cm,inner sep=0pt] (EE) at (8.1,-0.65) {\mbox{-}$\frac{\pi}{2}$};
    \draw[-] (E) -- (D);
    \draw[-] (EE) -- (DD);
    \node[shape=rectangle,draw=black,fill=white!20!yellow] (F) at (8.9,0.65) {};
    \node[shape=rectangle,draw=black,fill=white!20!yellow] (FF) at (8.9,-0.65) {};
    \draw[-] (E) -- (F);
    \draw[-] (EE) -- (FF);
    \node[shape=circle,draw=black,fill=white!80!green] (G) at (9.5,0.65) {};
    \node[shape=circle,draw=black,fill=white!50!red] (GGG) at (9.5,0) {};
    \node[shape=circle,draw=black,fill=white!80!green,minimum size = 0.5cm,inner sep=0pt] (GGGG) at (10.1,0) {$\theta$};
    \node[shape=circle,draw=black,fill=white!80!green] (GG) at (9.5,-0.65) {};
    \draw[-] (F) -- (G);
    \draw[-] (FF) -- (GG);
    \draw[-] (G) -- (GGG);
    \draw[-] (GG) -- (GGG);
    \draw[-] (GGGG) -- (GGG);
    \node[shape=rectangle,draw=black,fill=white!20!yellow] (H) at (10.1,0.65) {};
    \node[shape=rectangle,draw=black,fill=white!20!yellow] (HH) at (10.1,-0.65) {};
    \draw[-] (G) -- (H);
    \draw[-] (GG) -- (HH);
    \node[shape=circle,draw=black,fill=white!80!green] (I) at (10.85,0.65) {};
    \node[shape=circle,draw=black,fill=white!50!red] (III) at (10.85,0) {};
    \node[shape=circle,draw=black,fill=white!80!green,minimum size = 0.55cm,inner sep=0pt] (IIII) at (11.45,0) {$\frac{\phi}{2}$};
    \node[shape=circle,draw=black,fill=white!80!green] (II) at (10.85,-0.65) {};
    \draw[-] (H) -- (I);
    \draw[-] (HH) -- (II);
    \draw[-] (I) -- (III);
    \draw[-] (II) -- (III);
    \draw[-] (IIII) -- (III);
    \node[shape=circle,draw=black,fill=white!80!green,minimum size = 0.55cm,inner sep=0pt] (J) at (11.85,0.65) {\mbox{-}$\frac{\phi}{2}$};
    \node[shape=circle,draw=black,fill=white!80!green,minimum size = 0.55cm,inner sep=0pt] (JJ) at (11.85,-0.65) {\mbox{-}$\frac{\phi}{2}$};
    \draw[-] (J) -- (I);
    \draw[-] (JJ) -- (II);
    \coordinate (K) at (12.85,0.65);
    \coordinate (L) at (12.85,-0.65);
    \draw[-] (J) -- (K);
    \draw[-] (JJ) -- (L);
\end{tikzpicture}

\subsection{Graph-like diagrams}
\label{rules}
Consider the following rewriting rules, all valid in ZX-calculus:
\vspace{3mm}
\begin{enumerate}[font={\bfseries},label={Rule \arabic*.}]
\item

\begin{tikzpicture}[scale=0.6, transform shape]
    \node[shape=circle,draw=black,fill=white!80!green] (A) at (0,0) {$\alpha$};
    \coordinate (B) at (-1.5,0);
    \coordinate (C) at (1.5,0);
    \draw[-] (A) -- (B);
    \draw[-] (A) -- (C);
    \node[shape=rectangle,draw=black,fill=white!20!yellow] (C) at (-0.75,0) {};
    \node[shape=rectangle,draw=black,fill=white!20!yellow] (C) at (0.75,0) {};
    \node[shape=rectangle,draw=white] (eq) at (2.5,-0.05) {\huge $\Longleftrightarrow$};
    \node[shape=circle,draw=black,fill=white!50!red] (A2) at (4.5,0) {$\alpha$};
    \coordinate (B2) at (3.5,0);
    \coordinate (C2) at (5.5,0);
    \draw[-] (A2) -- (B2);
    \draw[-] (A2) -- (C2);
\end{tikzpicture}
 ,~
\begin{tikzpicture}[scale=0.6, transform shape]
    \node[shape=circle,draw=black,fill=white!50!red] (A) at (0,0) {$\alpha$};
    \coordinate (B) at (-1.5,0);
    \coordinate (C) at (1.5,0);
    \draw[-] (A) -- (B);
    \draw[-] (A) -- (C);
    \node[shape=rectangle,draw=black,fill=white!20!yellow] (C) at (-0.75,0) {};
    \node[shape=rectangle,draw=black,fill=white!20!yellow] (C) at (0.75,0) {};
    \node[shape=rectangle,draw=white] (eq) at (2.5,-0.05) {\huge $\Longleftrightarrow$};
    \node[shape=circle,draw=black,fill=white!80!green] (A2) at (4.5,0) {$\alpha$};
    \coordinate (B2) at (3.5,0);
    \coordinate (C2) at (5.5,0);
    \draw[-] (A2) -- (B2);
    \draw[-] (A2) -- (C2);
\end{tikzpicture}

Toggling from one spider type to another by creating Hadamard edges.

\vspace{7mm}
\item

\begin{tikzpicture}[scale=0.6, transform shape]
    \coordinate (A) at (-0.75,0);
    \coordinate (B) at (0.75,0);
    \draw[-] (A) -- (B);
    \node[shape=rectangle,draw=black,fill=white!20!yellow] (C) at (-0.25,0) {};
    \node[shape=rectangle,draw=black,fill=white!20!yellow] (C) at (0.25,0) {};
    \node[shape=rectangle,draw=white] (eq) at (2,-0.05) {\huge $\Longleftrightarrow$};
    \coordinate (A2) at (3.5,0);
    \coordinate (B2) at (5,0);
    \draw[-] (A2) -- (B2);
\end{tikzpicture}

Two Hadamards cancel each other.

\vspace{7mm}
\item
\label{rule:unfusion}
\begin{tikzpicture}[scale=0.6, transform shape]
    \node[shape=circle,draw=black,fill=white!80!green] (A) at (0,0) {$\alpha$};
    \node[shape=circle,draw=black,fill=white!80!green] (B) at (1.5,0) {$\beta$};
    \coordinate (C) at (-0.75,-0.5);
    \coordinate (D) at (-0.75,0.5);
    \coordinate (E) at (2.25,-0.5);
    \coordinate (F) at (2.25,0.5);
    \coordinate (C1) at (-0.6,-0.2);
    \coordinate (D1) at (-0.6,0.2);
    \coordinate (E1) at (2.1,-0.2);
    \coordinate (F1) at (2.1,0.2);
    \draw[-] (A) -- (B);
    \draw[-] (A) -- (C);
    \draw[-] (A) -- (D);
    \draw[-] (B) -- (E);
    \draw[-] (B) -- (F);
    \draw[loosely dotted] (C1) -- (D1);
    \draw[loosely dotted] (E1) -- (F1);
    \node[shape=rectangle,draw=white] (eq) at (3.15,-0.05) {\huge $\Longleftrightarrow$};
    \node[shape=circle,draw=black,fill=white!80!green] (A2) at (5.25,0) {$\alpha+\beta$};
    \coordinate (C2) at (4,-0.5);
    \coordinate (D2) at (4,0.5);
    \coordinate (E2) at (6.5,-0.5);
    \coordinate (F2) at (6.5,0.5);
    \coordinate (C3) at (4.15,-0.2);
    \coordinate (D3) at (4.15,0.2);
    \coordinate (E3) at (6.35,-0.2);
    \coordinate (F3) at (6.35,0.2);
    \draw[-] (A2) -- (C2);
    \draw[-] (A2) -- (D2);
    \draw[-] (A2) -- (E2);
    \draw[-] (A2) -- (F2);
    \draw[loosely dotted] (C3) -- (D3);
    \draw[loosely dotted] (E3) -- (F3);
\end{tikzpicture}

Spider fusion and unfusion.

\vspace{7mm}
\item

\begin{tikzpicture}[scale=0.6, transform shape]
    \coordinate (A) at (-0.75,0);
    \coordinate (B) at (0.75,0);
    \draw[-] (A) -- (B);
    \node[shape=circle,draw=black,fill=white!80!green] (A) at (0,0) {};
    \node[shape=rectangle,draw=white] (eq) at (2,-0.05) {\huge $\Longleftrightarrow$};
    \coordinate (A2) at (3.5,0);
    \coordinate (B2) at (5,0);
    \draw[-] (A2) -- (B2);
\end{tikzpicture}

Cancelling trivial rotations.
\end{enumerate}


Using this set of rules, we can transform any ZX-diagram into the following form.

\begin{definition}[Graph-like diagram] \label{def:graph_like}
A ZX-diagram is \textit{graph-like} if:
\begin{enumerate}
    \item All spiders are Z-spiders.
    \item All edges are Hadamard.
    \item There are no parallel Hadamard edges or self-loops.
    \item Every input or output is connected to a Z-spider and every Z-spider is connected to at most one input or output.
\end{enumerate}

\end{definition}

\begin{lemma}
Every ZX-diagram can be transformed into a graph-like ZX-diagram.
\end{lemma}
\begin{proof}
Starting with an arbitrary ZX-diagram, we can turn all red spiders into
green spiders surrounded by Hadamard gates using the first rule. We then remove excess Hadamards via the second rule. Any non-Hadamard edge is then removed by fusing the adjacent spiders with the third rule. We finally place trivial Z-spiders on every open wire that remain (see \cite{duncan2020graph} for a more in-depth demonstration).
\end{proof}

\begin{figure}[h!]
\centering
\begin{tikzpicture}[scale=0.75, transform shape]
    \node[shape=circle,draw=black,fill=white!50!red] (A) at (0,0) {$\alpha$};
    \node[shape=circle,draw=black,fill=white!50!red] (B) at (1,0) {};
    \node[shape=circle,draw=black,fill=white!80!green] (C) at (1,1) {};
    \node[shape=circle,draw=black,fill=white!80!green] (D) at (1,2) {$\beta$};
    \node[shape=circle,draw=black,fill=white!50!red] (E) at (0,1) {};
    \node[shape=circle,draw=black,fill=white!80!green] (F) at (0,2) {};
    \node[shape=circle,draw=black,fill=white!50!red] (G) at (2,1) {};
    \node[shape=circle,draw=black,fill=white!80!green] (H) at (2,2) {};
    \coordinate (F2) at (-1.25,2);
    \coordinate (E2) at (-1.25,1);
    \coordinate (A2) at (-1.25,0);
    \coordinate (H2) at (3.25,2);
    \coordinate (G2) at (3.25,1);
    \coordinate (B2) at (3.25,0);
    \draw[-] (B) -- (C);
    \draw[-] (A) -- (B);
    \draw[-] (F) -- (D);
    \draw[-] (E) -- (F);
    \draw[-] (E) -- (C);
    \draw[-] (G) -- (H);
    \draw[-] (C) -- (G);
    \draw[-] (D) -- (H);
    \draw[-] (F) -- (F2);
    \draw[-] (E) -- (E2);
    \draw[-] (A) -- (A2);
    \draw[-] (H) -- (H2);
    \draw[-] (G) -- (G2);
    \draw[-] (B) -- (B2);
    \node[shape=rectangle,draw=black,fill=white!20!yellow] (h) at (2.6,1) {};
    \node[shape=rectangle,draw=black,fill=white!20!yellow] (h) at (-0.7,2) {};
    \node[shape=rectangle,draw=black,fill=white!20!yellow] (h) at (-0.7,0) {};
    \node[shape=rectangle,draw=white] (eq) at (4,1) {\Large $\Longleftrightarrow$};
    
    \node[shape=circle,draw=black,fill=white!80!green] (B) at (7,0) {$\alpha$};
    \node[shape=circle,draw=black,fill=white!80!green] (C) at (7,1) {};
    \node[shape=circle,draw=black,fill=white!80!green] (D) at (7,2) {$\beta$};
    \node[shape=circle,draw=black,fill=white!80!green] (E) at (6,1) {};
    \node[shape=circle,draw=black,fill=white!80!green] (G) at (8,1) {};
    \coordinate (F2) at (4.75,2);
    \coordinate (E2) at (4.75,1);
    \coordinate (A2) at (4.75,0);
    \coordinate (H2) at (9.25,2);
    \coordinate (G2) at (9.25,1);
    \coordinate (B2) at (9.25,0);
    \draw[-] (B) -- (C);
    \draw[-] (E) -- (D);
    \draw[-] (E) -- (C);
    \draw[-] (G) -- (D);
    \draw[-] (C) -- (G);
    \draw[-] (D) -- (F2);
    \draw[-] (E) -- (E2);
    \draw[-] (B) -- (A2);
    \draw[-] (D) -- (H2);
    \draw[-] (G) -- (G2);
    \draw[-] (B) -- (B2);
    \node[shape=rectangle,draw=black,fill=white!20!yellow,rotate=45] (h) at (6.5,1.5) {};
    \node[shape=rectangle,draw=black,fill=white!20!yellow,rotate=45] (h) at (7.5,1.5) {};
    \node[shape=rectangle,draw=black,fill=white!20!yellow] (h) at (6,2) {};
    \node[shape=rectangle,draw=black,fill=white!20!yellow] (h) at (5.4,1) {};
    \node[shape=rectangle,draw=black,fill=white!20!yellow] (h) at (6.5,1) {};
    \node[shape=rectangle,draw=black,fill=white!20!yellow] (h) at (7.5,1) {};
    \node[shape=rectangle,draw=black,fill=white!20!yellow] (h) at (8,0) {};
    \node[shape=rectangle,draw=black,fill=white!20!yellow] (h) at (7,0.55) {}; 

    \node[shape=rectangle,draw=white] (eq) at (9.9,1) {\Large $\Longleftrightarrow$};

    \node[shape=circle,draw=black,fill=white!80!green] (B) at (13,0) {$\alpha$};
    \node[shape=circle,draw=black,fill=white!80!green] (C) at (13,1) {};
    \node[shape=circle,draw=black,fill=white!80!green] (D) at (13,2) {$\beta$};
    \node[shape=circle,draw=black,fill=white!80!green] (E) at (12,1) {};
    \node[shape=circle,draw=black,fill=white!80!green] (G) at (14,1) {};
    \node[shape=circle,draw=black,fill=white!80!green] (F2) at (11.75,2) {};
    \node[shape=circle,draw=black,fill=white!80!green] (E2) at (11,1) {};
    \node[shape=circle,draw=black,fill=white!80!green] (B2) at (14.25,0) {};
    \draw[-] (B) -- (C);
    \draw[-] (E) -- (D);
    \draw[-] (E) -- (C);
    \draw[-] (G) -- (D);
    \draw[-] (C) -- (G);
    \draw[-] (D) -- (F2);
    \draw[-] (E) -- (E2);
    \draw[-] (B) -- (B2);
    \draw[-] (F2) -- (11.4,2);
    \draw[-] (E2) -- (10.65,1);
    \draw[-] (B2) -- (14.6,0);
    \draw[-] (B) -- (12.55,0);
    \draw[-] (G) -- (14.35,1);
    \draw[-] (D) -- (13.55,2);

    \node[shape=rectangle,draw=black,fill=white!20!yellow,rotate=45] (h) at (12.5,1.5) {};
    \node[shape=rectangle,draw=black,fill=white!20!yellow,rotate=45] (h) at (13.5,1.5) {};
    \node[shape=rectangle,draw=black,fill=white!20!yellow] (h) at (12.25,2) {};
    \node[shape=rectangle,draw=black,fill=white!20!yellow] (h) at (11.5,1) {};
    \node[shape=rectangle,draw=black,fill=white!20!yellow] (h) at (12.5,1) {};
    \node[shape=rectangle,draw=black,fill=white!20!yellow] (h) at (13.5,1) {};
    \node[shape=rectangle,draw=black,fill=white!20!yellow] (h) at (13.65,0) {};
    \node[shape=rectangle,draw=black,fill=white!20!yellow] (h) at (13,0.55) {}; 
\end{tikzpicture}
\caption{Rewriting a ZX-diagram into a graph-like diagram.}
\label{fig:diagram_to_graph_like}
\end{figure}
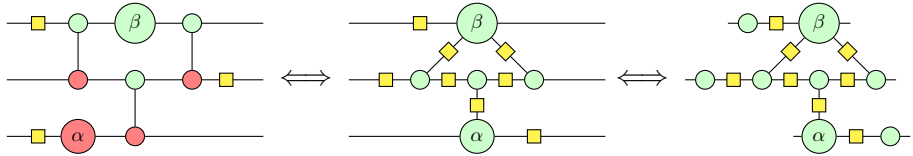

Since we are interested in computing a quantity of shape $\bra{0}C\ket{0}$, the resulting diagram will have no input or output. Thus, in our setting, property $4$ of Definition \ref{def:graph_like} is always met.
We say that such graph-like diagrams are \textit{closed}.


\begin{figure}[h!]
\centering
\begin{tikzpicture}[scale=0.75, transform shape]
    \node[shape=circle,draw=black,fill=white!50!red] (A) at (0,0) {$\alpha$};
    \node[shape=circle,draw=black,fill=white!50!red,minimum size = 0.2cm,inner sep=0pt] (aa) at (-1,0) {};
    \node[shape=circle,draw=black,fill=white!50!red] (B) at (1,0) {};
    \node[shape=circle,draw=black,fill=white!50!red,minimum size = 0.2cm,inner sep=0pt] (bb) at (3,0) {};
    \node[shape=circle,draw=black,fill=white!80!green] (C) at (1,1) {};
    \node[shape=circle,draw=black,fill=white!80!green] (D) at (1,2) {$\beta$};
    \node[shape=circle,draw=black,fill=white!50!red] (E) at (0,1) {};
    \node[shape=circle,draw=black,fill=white!50!red,minimum size = 0.2cm,inner sep=0pt] (ee) at (-1,1) {};
    \node[shape=circle,draw=black,fill=white!80!green] (F) at (0,2) {};
    \node[shape=circle,draw=black,fill=white!50!red,minimum size = 0.2cm,inner sep=0pt] (ff) at (-1,2) {};
    \node[shape=circle,draw=black,fill=white!50!red] (G) at (2,1) {};
    \node[shape=circle,draw=black,fill=white!50!red,minimum size = 0.2cm,inner sep=0pt] (gg) at (3,1) {};
    \node[shape=circle,draw=black,fill=white!80!green] (H) at (2,2) {};
    \node[shape=circle,draw=black,fill=white!50!red,minimum size = 0.2cm,inner sep=0pt] (hh) at (3,2) {};
    \draw[-] (B) -- (C);
    \draw[-] (A) -- (B);
    \draw[-] (F) -- (D);
    \draw[-] (E) -- (F);
    \draw[-] (E) -- (C);
    \draw[-] (G) -- (H);
    \draw[-] (C) -- (G);
    \draw[-] (D) -- (H);
    \draw[-] (A) -- (aa);
    \draw[-] (B) -- (bb);
    \draw[-] (E) -- (ee);
    \draw[-] (F) -- (ff);
    \draw[-] (G) -- (gg);
    \draw[-] (H) -- (hh);
    \node[shape=rectangle,draw=black,fill=white!20!yellow] (h) at (2.5,1) {};
    \node[shape=rectangle,draw=black,fill=white!20!yellow] (h) at (-0.5,2) {};
    \node[shape=rectangle,draw=black,fill=white!20!yellow] (h) at (-0.6,0) {};
    \node[shape=rectangle,draw=white] (eq) at (4,1) {\Large $\Longleftrightarrow$};
    
    \node[shape=circle,draw=black,fill=white!80!green] (B) at (6.2,0) {$\alpha$};
    \node[shape=circle,draw=black,fill=white!80!green] (C) at (6.2,1) {};
    \node[shape=circle,draw=black,fill=white!80!green] (D) at (6.2,2) {$\beta$};
    \node[shape=circle,draw=black,fill=white!80!green] (E) at (5.2,1) {};
    \node[shape=circle,draw=black,fill=white!80!green] (G) at (7.2,1) {};
    \node[shape=circle,draw=black,fill=white!80!green] (F2) at (7.45,2) {};
    \node[shape=circle,draw=black,fill=white!80!green] (E2) at (8.2,1) {};
    \node[shape=circle,draw=black,fill=white!80!green] (B2) at (4.95,0) {};
    \draw[-] (B) -- (C);
    \draw[-] (E) -- (D);
    \draw[-] (E) -- (C);
    \draw[-] (G) -- (D);
    \draw[-] (C) -- (G);
    \draw[-] (D) -- (F2);
    \draw[-] (G) -- (E2);
    \draw[-] (B) -- (B2);
    \node[shape=rectangle,draw=black,fill=white!20!yellow,rotate=45] (h) at (5.7,1.5) {};
    \node[shape=rectangle,draw=black,fill=white!20!yellow,rotate=45] (h) at (6.7,1.5) {};
    \node[shape=rectangle,draw=black,fill=white!20!yellow] (h) at (6.9,2) {};
    \node[shape=rectangle,draw=black,fill=white!20!yellow] (h) at (7.7,1) {};
    \node[shape=rectangle,draw=black,fill=white!20!yellow] (h) at (5.7,1) {};
    \node[shape=rectangle,draw=black,fill=white!20!yellow] (h) at (6.7,1) {};
    \node[shape=rectangle,draw=black,fill=white!20!yellow] (h) at (5.55,0) {};
    \node[shape=rectangle,draw=black,fill=white!20!yellow] (h) at (6.2,0.55) {}; 

    \node[shape=rectangle,draw=white] (eq) at (9.125,1) {\Large $\Longleftrightarrow$};
    
    \coordinate (B) at (11.2,0);
    \coordinate (C) at (11.2,1);
    \coordinate (D) at (11.2,2);
    \coordinate (E) at (10.2,1);
    \coordinate (G) at (12.2,1);
    \coordinate (F2) at (12.45,2);
    \coordinate (E2) at (13.2,1);
    \coordinate (B2) at (9.95,0);
    \draw[-] (B) -- (C);
    \draw[-] (E) -- (D);
    \draw[-] (E) -- (C);
    \draw[-] (G) -- (D);
    \draw[-] (C) -- (G);
    \draw[-] (D) -- (F2);
    \draw[-] (G) -- (E2);
    \draw[-] (B) -- (B2);
    \filldraw [black] (B) circle (2pt);
    \filldraw [black] (C) circle (2pt);
    \filldraw [black] (D) circle (2pt);
    \filldraw [black] (E) circle (2pt);
    \filldraw [black] (G) circle (2pt);
    \filldraw [black] (F2) circle (2pt);
    \filldraw [black] (E2) circle (2pt);
    \filldraw [black] (B2) circle (2pt);
\end{tikzpicture}
\caption{Rewriting a closed ZX-diagram into a closed graph-like.}\label{fig:diagram_to_graph}
\end{figure}
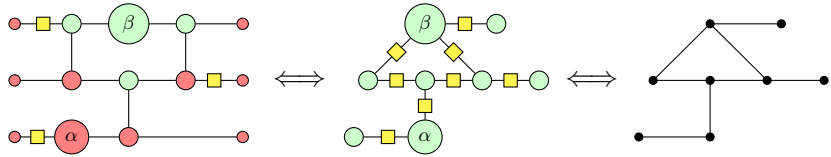

This resulting graph will essentially be our tensor network (we will often call it graph-like diagram with no distinction between ZX-calculus and graph theory). We can now start our simplification process: we will rewrite our graph-like diagram via some ZX-rewriting in order to sparsify it.

\subsection{Rewriting graph-like diagrams}
We now introduce different rewriting rules that we will use in our simplification strategy.

\subsubsection{Local complementations and pivots}
A local complementation is a graph minor introduced by Kotzig \cite{kotzig1968eulerian} (see Figure \ref{fig:loc_comp}):

\begin{definition}[Local complementation]
Let $G$ be a graph and let $u$ be a vertex of $G=(V,E)$.
The \textit{local complementation} of $G$ according to $u$, written  $G\star u$, is a graph which has the same vertices as $G$, but all the neighbours $v_1,v_2,\dots$ of $u$ are connected in $G\star u$ if and only if they are not connected in $G$. All other edges are unchanged.

\end{definition}
\begin{figure}[h!]
\centering
\begin{tikzpicture}[scale=0.9, transform shape]
    \node[shape=circle,draw=black,fill=black] (A) at (0,0) {};
    \node[] (u) at (0,-0.5) {\Large $u$};
    \node[shape=circle,draw=black,fill=black] (B) at (1,0.65) {};
    \node[] (v1) at (1.5,0.65) {\Large $v_1$};
    \node[shape=circle,draw=black,fill=black] (C) at (1.5,-0.6) {};
    \node[] (V2) at (2,-0.6) {\Large $v_2$};
    \node[shape=circle,draw=black,fill=black] (D) at (2.5,-1.5) {};
    \node[] (v3) at (3,-1.5) {\Large $v_3$};
    \draw[-] (A) -- (B);
    \draw[-] (A) -- (C);
    \draw[-] (A) -- (D);
    \draw[-] (B) -- (C);
    \draw[-] (C) -- (D);
    \node[] (eq) at (3.2,0) {\huge $\longrightarrow$};
    \node[] (eq) at (3.2,0.5) {\Large $G\star u$};
    
    \node[shape=circle,draw=black,fill=black] (A) at (5,0) {};
    \node[] (u) at (5,-0.5) {\Large $u$};
    \node[shape=circle,draw=black,fill=black] (B) at (6,0.65) {};
    \node[] (v1) at (6.5,0.65) {\Large $v_1$};
    \node[shape=circle,draw=black,fill=black] (C) at (6.5,-0.6) {};
    \node[] (v2) at (7.1,-0.5) {\Large $v_2$};
    \node[shape=circle,draw=black,fill=black] (D) at (7.5,-1.5) {};
    \node[] (v3) at (8,-1.5) {\Large $v_3$};
    \draw[-] (A) -- (B);
    \draw[-] (A) -- (C);
    \draw[-] (A) -- (D);
    \draw[-] (B) -- (D);
    \draw[dashed,color=white!80!blue] (B) -- (C);
    \draw[dashed,color=white!80!blue] (C) -- (D);
    
\end{tikzpicture}
\caption{Local complementation of a graph $G$ according to $u$. $v_1$ and $v_3$ are connected in $G\star u$ since the edge $v_1v_3\notin E$, while $v_1v_2\in E$ and $v_2v_3\in E$.}
\label{fig:loc_comp}
\end{figure}
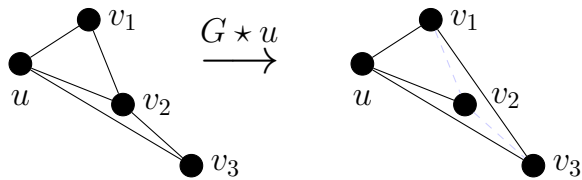

It so happens that we can perform this graph transformation on a particular set of graph-like diagrams called \textit{graph states}. Graph states are graph-like diagrams where all spider angles are $0$ and all spiders have an output (i.e. a dangling edge).
Given some graph state $G$ and one of its spider $u$, one can perform a local complementation $G\star u$ and fix the graph state by:
\begin{itemize}
    \item adding a red spider of angle $-\pi/2$ to the output of $u$
    \item adding green spiders of angle $\pi/2$ to the outputs of the neighbors of $u$
\end{itemize}
Figure \ref{fig:loc_comp_graph_state} depicts a graph-state local complementation.
This definition is standard in the ZX-calculus, see \cite{duncan2020graph} for a thorough proof.

\begin{figure}[h!]
\centering
\begin{tikzpicture}[scale=0.9, transform shape]
    \node[shape=circle,draw=black,fill=white!80!green] (A) at (0,0) {};
    \node[shape=rectangle,draw=white] (u) at (0,-0.5) {\Large $u$};
    \node[shape=circle,draw=black,fill=white!80!green] (B) at (1,0.65) {};
    \node[shape=circle,draw=black,fill=white!80!green] (C) at (1.5,-0.6) {};
    \node[shape=circle,draw=black,fill=white!80!green] (D) at (2.5,-1.5) {};
    \coordinate (A2) at (0,1.1);
    \coordinate (B2) at (1,1.7);
    \coordinate (C2) at (1.5,1.35);
    \coordinate (D2) at (2.5,1.1);
    \draw[-,color=white!30!blue] (A) -- (B);
    \draw[-,color=white!30!blue] (A) -- (C);
    \draw[-,color=white!30!blue] (A) -- (D);
    \draw[-,color=white!30!blue] (B) -- (C);
    \draw[-,color=white!30!blue] (C) -- (D);
    \draw[-] (A) -- (A2);
    \draw[-] (B) -- (B2);
    \draw[-] (C) -- (C2);
    \draw[-] (D) -- (D2);
    \draw[loosely dotted] (1.7,0.15) -- (2.3,0.15);
    \node[] (eq) at (3.6,0) {\huge $\longrightarrow$};
    \node[] (eq) at (3.6,0.5) {\Large $G\star u$};
    
    \node[shape=circle,draw=black,fill=white!80!green] (A) at (5,0) {};
    \node[shape=rectangle,draw=white] (u) at (5,-0.5) {\Large $u$};
    \node[shape=circle,draw=black,fill=white!80!green] (B) at (6,0.65) {};
    \node[shape=circle,draw=black,fill=white!80!green] (C) at (6.5,-0.6) {};
    \node[shape=circle,draw=black,fill=white!80!green] (D) at (7.5,-1.5) {};
    \coordinate (A2) at (5,1.1);
    \coordinate (B2) at (6,1.7);
    \coordinate (C2) at (6.5,1.35);
    \coordinate (D2) at (7.5,1.1);
    \draw[-,color=white!30!blue] (A) -- (B);
    \draw[-,color=white!30!blue] (A) -- (C);
    \draw[-,color=white!30!blue] (A) -- (D);
    \draw[-,color=white!30!blue] (B) -- (D);
    \draw[dashed,color=white!80!blue] (B) -- (C);
    \draw[dashed,color=white!80!blue] (C) -- (D);
    \draw[-] (A) -- (A2);
    \draw[-] (B) -- (B2);
    \draw[-] (C) -- (C2);
    \draw[-] (D) -- (D2);
    \draw[loosely dotted] (6.7,0.15) -- (7.3,0.15);
    \node[shape=circle,draw=black,fill=white!50!red,minimum size=5mm] (A) at (5,0.6) {};
    \node[shape=rectangle] (a) at (5,0.6) {\mbox{-}$\frac{\pi}{2}$};
    \node[shape=circle,draw=black,fill=white!80!green,minimum size=5mm] (B) at (6,1.25) {};
    \node[shape=rectangle] (b) at (6,1.25) {$\frac{\pi}{2}$};
    \node[shape=circle,draw=black,fill=white!80!green,minimum size=5mm] (C) at (6.5,0.85) {};
    \node[shape=rectangle] (c) at (6.5,0.85) {$\frac{\pi}{2}$};
    \node[shape=circle,draw=black,fill=white!80!green,minimum size=5mm] (D) at (7.5,0.6) {};
    \node[shape=rectangle] (d) at (7.5,0.6) {$\frac{\pi}{2}$};
    
\end{tikzpicture}
\caption{ZX-diagram representation of a local complementation on a graph-state $G$ according to $u$.}
\label{fig:loc_comp_graph_state}
\end{figure}
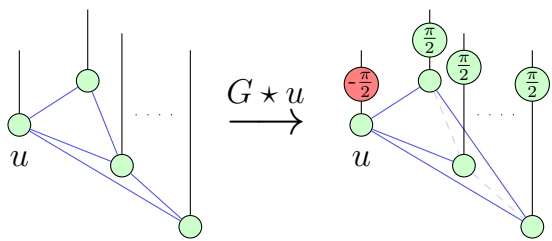

However, we want to apply local complementations to arbitrary closed graph-like diagrams (i.e. graph states with no outputs and non-zero angles).
We propose the following local complementation procedure for closed graph-like diagrams (see Figure \ref{fig:loc_comp_graph_like}):
\begin{itemize}
    \item start by pulling the angles out of the spiders via the unfusion rule (\ref{rule:unfusion}),
    \item apply a graph state local complementation on the resulting graph state 
\end{itemize}

Notice that the resulting diagram is not graph-like, as it contains an additional red spider and non-Hadamard edges.
Notice also that applying many such transformations to a graph-like will ``grow" chains of degree $2$ spiders in between the degree $1$ $Z$-spiders carrying the angles and the spiders of the graph state. However, up to those chains, the resulting diagram has the same structure as its graph state portion.

Those chains can be seen as tensor networks implementing simple linear forms (or vectors) of dimension $2$. It is simple to update their value every time a local complementation is performed on the graph-like diagram:
\begin{itemize}
    \item their value can be initially set to $\left[1, e^{i\theta}\right]$ where $\theta$ is the angle of initial spider
    \item then, for every local complementation on some spider $u$:
    \begin{itemize}
        \item update the linear form $L_u$ of $u$ via: $$L_u\gets R_X(-\pi/2)L_u$$
        \item and for each neighbor $v\in N(u)$, update its linear form $L_v$ via: $$L_v \gets R_Z(\pi/2)L_v$$
    \end{itemize}
\end{itemize}
The resulting data structure is somewhat hybrid, in the sense that it describes a graph state along with some additional linear forms attached to each node of the graph state. Hence it is not a ZX-diagram per se.

Notice that this bookkeeping is overall cheaper than the local complementation itself: we simply need to update $|N(u)|+1$ constant size linear forms, while we need to modify $|N(u)|^2$ adjacency relations to update the graph state.


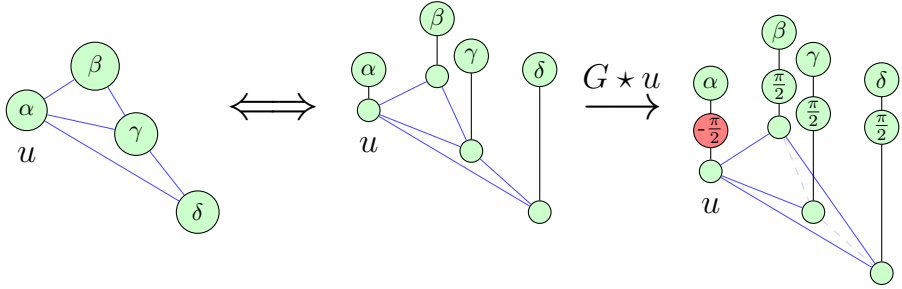
\begin{figure}[h!]
\centering
\begin{tikzpicture}[scale=0.9, transform shape]
    \node[shape=circle,draw=black,fill=white!80!green] (A) at (0,0) {$\alpha$};
    \node[shape=rectangle,draw=white] (u) at (0,-0.65) {\Large $u$};
    \node[shape=circle,draw=black,fill=white!80!green] (B) at (1,0.65) {$\beta$};
    \node[shape=circle,draw=black,fill=white!80!green] (C) at (1.6,-0.35) {$\gamma$};
    \node[shape=circle,draw=black,fill=white!80!green] (D) at (2.5,-1.5) {$\delta$};
    \coordinate (A2) at (0,1.1);
    \coordinate (B2) at (1,1.7);
    \coordinate (C2) at (1.5,1.35);
    \coordinate (D2) at (2.5,1.1);
    \draw[-,color=white!30!blue] (A) -- (B);
    \draw[-,color=white!30!blue] (A) -- (C);
    \draw[-,color=white!30!blue] (A) -- (D);
    \draw[-,color=white!30!blue] (B) -- (C);
    \draw[-,color=white!30!blue] (C) -- (D);
    \node[shape=rectangle,draw=white] (eq) at (3.6,0) {\huge $\Longleftrightarrow$};
    
    \node[shape=circle,draw=black,fill=white!80!green] (A) at (5,0) {};
    \node[shape=rectangle,draw=white] (u) at (5,-0.5) {\Large $u$};
    \node[shape=circle,draw=black,fill=white!80!green] (B) at (6,0.5) {};
    \node[shape=circle,draw=black,fill=white!80!green] (C) at (6.5,-0.6) {};
    \node[shape=circle,draw=black,fill=white!80!green] (D) at (7.5,-1.5) {};
    \coordinate (A1) at (5,0.6);
    \coordinate (B1) at (6,1.35);
    \coordinate (C1) at (6.5,0.8);
    \coordinate (D1) at (7.5,0.6);
    \coordinate (A2) at (5,1.2);
    \coordinate (B2) at (6,1.95);
    \coordinate (C2) at (6.5,1.4);
    \coordinate (D2) at (7.5,1.2);
    \draw[-,color=white!30!blue] (A) -- (B);
    \draw[-,color=white!30!blue] (A) -- (C);
    \draw[-,color=white!30!blue] (A) -- (D);
    \draw[-,color=white!30!blue] (B) -- (C);
    \draw[-,color=white!30!blue] (C) -- (D);
    \draw[-] (A) -- (A1);
    \draw[-] (B) -- (B1);
    \draw[-] (C) -- (C1);
    \draw[-] (D) -- (D1);
    \node[shape=circle,draw=black,fill=white!80!green,minimum size=5mm] (A) at (5,0.6) {};
    \node[shape=rectangle] (a) at (5,0.6) {$\alpha$};
    \node[shape=circle,draw=black,fill=white!80!green,minimum size=5mm] (B) at (6,1.35) {};
    \node[shape=rectangle] (b) at (6,1.35) {$\beta$};
    \node[shape=circle,draw=black,fill=white!80!green,minimum size=5mm] (C) at (6.5,0.8) {};
    \node[shape=rectangle] (c) at (6.5,0.8) {$\gamma$};
    \node[shape=circle,draw=black,fill=white!80!green,minimum size=5mm] (D) at (7.5,0.6) {};
    \node[shape=rectangle] (d) at (7.5,0.6) {$\delta$};
    \node[] (eq) at (8.7,0) {\huge $\longrightarrow$};
    \node[] (eq) at (8.7,0.5) {\Large $G\star u$};
    
    \node[shape=circle,draw=black,fill=white!80!green] (A) at (10,-0.9) {};
    \node[shape=rectangle,draw=white] (u) at (10,-1.4) {\Large $u$};
    \node[shape=circle,draw=black,fill=white!80!green] (B) at (11,-0.25) {};
    \node[shape=circle,draw=black,fill=white!80!green] (C) at (11.5,-1.5) {};
    \node[shape=circle,draw=black,fill=white!80!green] (D) at (12.5,-2.4) {};
    \coordinate (A1) at (10,0.4);
    \coordinate (B1) at (11,1.15);
    \coordinate (C1) at (11.5,0.75);
    \coordinate (D1) at (12.5,0.45);
    \coordinate (A2) at (10,0.95);
    \coordinate (B2) at (11,1.75);
    \coordinate (C2) at (11.5,1.35);
    \coordinate (D2) at (12.5,1.05);
    \draw[-,color=white!30!blue] (A) -- (B);
    \draw[-,color=white!30!blue] (A) -- (C);
    \draw[-,color=white!30!blue] (A) -- (D);
    \draw[-,color=white!30!blue] (B) -- (D);
    \draw[dashed,color=white!80!blue] (B) -- (C);
    \draw[dashed,color=white!80!blue] (C) -- (D);
    \draw[-] (A) -- (A1);
    \draw[-] (B) -- (B1);
    \draw[-] (C) -- (C1);
    \draw[-] (D) -- (D1);
    \node[shape=circle,draw=black,fill=white!80!green,minimum size=5mm] (A) at (10,0.4) {};
    \node[shape=rectangle] (a) at (10,0.4) {$\alpha$};
    \node[shape=circle,draw=black,fill=white!80!green,minimum size=5mm] (B) at (11,1.15) {};
    \node[shape=rectangle] (b) at (11,1.15) {$\beta$};
    \node[shape=circle,draw=black,fill=white!80!green,minimum size=5mm] (C) at (11.5,0.75) {};
    \node[shape=rectangle] (c) at (11.5,0.75) {$\gamma$};
    \node[shape=circle,draw=black,fill=white!80!green,minimum size=5mm] (D) at (12.5,0.45) {};
    \node[shape=rectangle] (d) at (12.5,0.45) {$\delta$};
    \node[shape=circle,draw=black,fill=white!50!red,minimum size=5mm] (A) at (10,-0.3) {};
    \node[shape=rectangle] (a) at (10,-0.3) {\mbox{-}$\frac{\pi}{2}$};
    \node[shape=circle,draw=black,fill=white!80!green,minimum size=5mm] (B) at (11,0.35) {};
    \node[shape=rectangle] (b) at (11,0.35) {$\frac{\pi}{2}$};
    \node[shape=circle,draw=black,fill=white!80!green,minimum size=5mm] (C) at (11.5,-0.05) {};
    \node[shape=rectangle] (c) at (11.5,-0.05) {$\frac{\pi}{2}$};
    \node[shape=circle,draw=black,fill=white!80!green,minimum size=5mm] (D) at (12.5,-0.25) {};
    \node[shape=rectangle] (d) at (12.5,-0.25) {$\frac{\pi}{2}$};
    
\end{tikzpicture}
\caption{ZX-diagram of local complementation on a closed graph-like $G$ according to $u$, after pulling the angles out of the spiders.}
\label{fig:loc_comp_graph_like}
\end{figure}

\tikzstyle{gn}=[rectangle,rounded corners=0.8em,fill=zxgreen,draw=Black,
  line width=0.8 pt,inner sep=3pt,minimum width=1.5em,minimum height=1.5em]
\tikzstyle{rn}=[rectangle,rounded corners=0.8em,fill=zxred,draw=Black,
  line width=0.8 pt,inner sep=3pt,minimum width=1.5em,minimum height=1.5em]
\definecolor{zxred}{RGB}{232, 165, 165}
\definecolor{zxgreen}{RGB}{216, 248, 216}
\tikzstyle{nm}=[draw, circle,fill=black, inner sep=1pt]

We will use the following local transformation based on local complementation:

\begin{definition}[Pivot]
Let $G$ be a graph and let $u$ and $v$ be a pair of connected vertices in $G$.
The \textit{pivot} of $G$ along the edge $uv$ is the graph $G\wedge uv := G\star u\star v\star u$.
\end{definition}

\begin{figure}[h!]
\centering
\begin{tikzpicture}
    \draw (0.25, 1) node[nm](u){};
    \draw (1.75, 1) node[nm](v){};
    \draw (0, 0) node[nm](b){};
    \draw (1, 0.5) node[nm](a){};
    \draw (2, 0) node[nm](c){};
    \node[shape=rectangle,draw=white] (U) at (-0.1,1) {\large $u$};
    \node[shape=rectangle,draw=white] (V) at (2.1,1) {\large $v$};
    \draw (u) -- (v);
    \draw (u) -- (b);
    \draw (u) -- (a);
    \draw (v) -- (a);
    \draw (v) -- (c);
    \draw (b) -- (c);
    \draw (0,-0.1) -- (2,-0.1);
    \draw (b) -- (1,0.4);
    \draw (0,0.1) -- (a);
    \draw (c) -- (1,0.4);
    \draw (2,0.1) -- (a);
    \node[shape=circle,draw=black,fill=white] (A) at (1,0.5) {$A$};
    \node[shape=circle,draw=black,fill=white] (B) at (0,0) {$B$};
    \node[shape=circle,draw=black,fill=white] (C) at (2,0) {$C$};

    \node[] (eq) at (3,0.3) {\Large $\longrightarrow$};
    \node[] (eq) at (3,0.7) {$G\wedge uv$};

    \draw (4.25, 1) node[nm](u){};
    \draw (5.75, 1) node[nm](v){};
    \draw (4, 0) node[nm](b){};
    \draw (5, 0.5) node[nm](a){};
    \draw (6, 0) node[nm](c){};
    \node[shape=rectangle,draw=white] (U) at (3.9,1) {\large $v$};
    \node[shape=rectangle,draw=white] (V) at (6.1,1) {\large $u$};
    \draw (u) -- (v);
    \draw (u) -- (b);
    \draw (u) -- (a);
    \draw (v) -- (a);
    \draw (v) -- (c);
    
    \draw (4,-0.125) -- (6, 0.025);
    \draw (4, 0.025) -- (6,-0.125);

    \draw (4,0.2) -- (5,0.3);
    \draw (4, -0.1) -- (5, 0.6);
    
    \draw (6,0.2) -- (5,0.3);
    \draw (6, -0.1) -- (5, 0.6);
    \node[shape=circle,draw=black,fill=white] (A) at (5,0.5) {$A$};
    \node[shape=circle,draw=black,fill=white] (B) at (4,0) {$B$};
    \node[shape=circle,draw=black,fill=white] (C) at (6,0) {$C$};

    \node[shape=rectangle,draw=white] at (3,-1) {where $A = N_G(u)\cap N_G(v)$, $B = N_G(u)$ and $C = N_G(v)$.};
\end{tikzpicture}
\caption{Pivot of $G$ along the edge $uv$: toggles edges between the common neighborhood of $u$ and $v$, the exclusive neighborhood of $u$ and the exclusive neighborhood of $v$, additionally swaps $u$ and $v$.}
\label{fig:pivot}
\end{figure}

\subsubsection{Spider unfusion}

For a vertex $u$ in $G$, using the \textit{unfusion rule} introduced in Section \ref{rules} on $u$ produces a new node $u'$. However, it does not preserve the graph-likeness of a ZX-diagram because of the non-Hadamard edge it creates.

\begin{definition}[Graph-like unfusion rule]
Let $G$ be a graph-like diagram and let $u$ be a spider of $G$. After unfusioning $u$ into two spiders $u, u'$, insert a new trivial X-spider $u''$ on the wire $uu'$, then toggle it to a Z-spider thus creating Hadamard edges.
\end{definition}

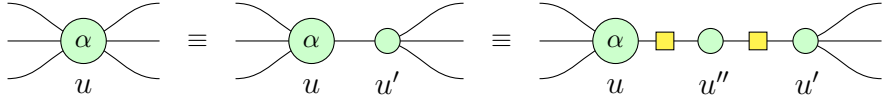
\begin{figure}[h!]
\centering
\begin{tikzpicture}
    
    \draw (0, 0) node[nm](n0){};
    \node[shape=rectangle,draw=white] (u) at (0,-0.6) {\large $u$};
    \draw (-1, 0.5) to[out=0,in=160] (n0);
    \draw (-1, 0) -- (n0);
    \draw (-1, -0.5) to[out=0,in=200] (n0);
    \draw (1, 0.5) to[in=20,out=180] (n0);
    \draw (1, 0) -- (n0);
    \draw (1, -0.5) to[in=-20,out=180] (n0);
    \node[shape=circle,draw=black,fill=white!80!green] (A) at (0,0) {$\alpha$};
    
    \draw (1.5, 0) node{$\equiv$};
    \draw (3, 0) node[nm](n1){};
    \node[shape=rectangle,draw=white] (u) at (3,-0.6) {\large $u$};
    \draw (2, 0.5) to[out=0,in=160] (n1);
    \draw (2, 0) -- (n1);
    \draw (2, -0.5) to[out=0,in=200] (n1);
    \draw (4, 0) node[nm](n3){};
    \node[shape=rectangle,draw=white] (u) at (4,-0.55) {\large $u'$};
    \draw (n1) -- (n3);
    \draw (5, 0.5) to[in=20,out=180] (n3);
    \draw (5, 0) -- (n3);
    \draw (5, -0.5) to[in=-20,out=180] (n3);
    \node[shape=circle,draw=black,fill=white!80!green] (A) at (3,0) {$\alpha$};
    \node[shape=circle,draw=black,fill=white!80!green] (A) at (4,0) {};

    \draw (5.5, 0) node{$\equiv$};
    \draw (7, 0) node[nm](n1){};
    \node[shape=rectangle,draw=white] (u) at (7,-0.6) {\large $u$};
    \draw (6, 0.5) to[out=0,in=160] (n1);
    \draw (6, 0) -- (n1);
    \draw (6, -0.5) to[out=0,in=200] (n1);
    \draw (8.25, 0) node[nm](n3){};
    \node[shape=rectangle,draw=white] (u) at (8.3,-0.55) {\large $u''$};
    \draw (9.5, 0) node[nm](n2){};
    \node[shape=rectangle,draw=white] (u) at (9.525,-0.55) {\large $u'$};
    \draw (n1) -- (n3);
    \draw (n2) -- (n3);
    \draw (10.5, 0.5) to[in=20,out=180] (n2);
    \draw (10.5, 0) -- (n2);
    \draw (10.5, -0.5) to[in=-20,out=180] (n2);
    \node[shape=rectangle,draw=black,fill=white!20!yellow] (h) at (7.65,0) {};
    \node[shape=rectangle,draw=black,fill=white!20!yellow] (h) at (8.875,0) {};
    \node[shape=circle,draw=black,fill=white!80!green] (A) at (7,0) {$\alpha$};
    \node[shape=circle,draw=black,fill=white!80!green] (A) at (8.25,0) {};
    \node[shape=circle,draw=black,fill=white!80!green] (A) at (9.5,0) {};
    
\end{tikzpicture}
\caption{Graph-like unfusion rule.}
\label{fig:unfusion}
\end{figure}

Notice that since each node in our graph-like carries a linear form attach to it, we need to decide what it becomes after the splitting. In the diagram of Figure \ref{fig:unfusion}, we assume that the linear form is attached to the node carrying the angle (here $u$). In practice, this means that the two new nodes $u'$ and $u''$ will receive a trivial linear form $\sqrt{2}\cdot\bra{+}=(1\ 1)$ and $u$'s original linear form will remain attached to $u$.

\begin{lemma}
This unfusion rule preserves the graph-likeness of a ZX-diagram.
\end{lemma}
\begin{proof}
    Use the identity rule to remove $u''$, followed by Hadamard cancellation and spider fusion.
\end{proof}

Our strategy uses the graph-like unfusion rule to \textit{split} the graph - i.e. to recursively unfusion all high degrees nodes of the graph until all spiders created have a degree less than or equal to 3. 
However, when performing an unfusion, we need to first partition the legs of $u$ into two sets: one set will remain attached to $u$, while the other will be attached to $u'$. We use a in-house heuristic based on computing a cycle basis of $G$ around $u$ (see Section \ref{appendix} for more details) in order to build this partition.

\section{Simplification strategies}
\label{strat}
\subsection{Our strategy}
Our simplification strategy consists in performing random pivots on a graph-like diagram and giving the resulting diagram a score correlated to its contraction cost. We used this strategy in a simulated annealing setting [KIR82] with parameters described in Subsection \ref{param}. Upon reaching a temperature of 0, this algorithm is equivalent to a greedy algorithm. We tested both modes of operation in our benchmarks.

\smallskip

Thus, our method consists of the following steps:
\begin{itemize}
    \item translate the input quantum circuit and the projector $\bra{x}$ into a ZX-diagram,
    \item transform the ZX-diagram into a graph-like diagram,
    \item optimize the graph-like diagram using ZX rewriting rules and pivots that lower the treewidth approximation given by QuickBB, then split all the high-degree nodes,
    \item pre-contract low degree nodes into their neighbors as explained in Subsection \ref{order},
    \item carry on with the state of the art method in order to compute a contraction order as described in Subsection \ref{order},
    \item contract the network.
\end{itemize}

\begin{algorithm}
\caption{Simulated annealing for graph-like optimization}
\label{algo2}
\begin{algorithmic}
\Require{graph-like diagram, number of steps, cost function} 
\Ensure{optimized graph-like diagram}
\Statex
\Function{SimulatedAnnealing}{$graphlike$, $nb\_steps$, $\textsc{CostFunction}$}
  \State {solution$_{current}\gets graphlike$}
  \State {energy$_{current}\gets$ \Call{CostFunction}{solution$_{current}$}}
  \State {energy$_{optimal}\gets$ energy$_{current}$} \Comment{Initialization}
  \State {$step\gets 0$}
  \While{$step < nb\_steps$}
    \State { $\triangleright$ Two graph-like are neighbors if they differ by one pivot}
    \State{solution$_{candidate}\gets$\Call{RandomPivot}{solution$_{current}$}} 
    \State {energy$_{candidate}\gets$ \Call{CostFunction}{solution$_{candidate}$}}
    \If{energy$_{candidate} <$ energy$_{optimal}$}
      \State{solution$_{optimal}\gets$ solution$_{candidate}$}
      \State{energy$_{optimal}\gets$ energy$_{candidate}$}
    \EndIf
    \Statex
    \State {$\tau\gets$\Call{Temperature}{$step / nb\_steps$}}
    \State {$prob\gets$\Call{Accept}{energy$_{current}$, energy$_{candidate}$, $\tau$}}
    \Statex
    \If{$prob >$\Call{Random}{$0$,$1$}}\Comment{Uniformly picks a random float between 0 and 1}
      \State{solution$_{current}\gets$ solution$_{candidate}$}
      \State{energy$_{current}\gets$ energy$_{candidate}$}
    \EndIf
  \EndWhile
  \State \Return{solution$_{optimal}$}
  \State { $\triangleright$ Definitions of $\textsc{CostFunction}$, $\textsc{Temperature}$ and $\textsc{Accept}$ in Section \ref{appendix}}
\EndFunction
\end{algorithmic}
\end{algorithm}

\subsection{Optimization parameters}
\label{param}
Our contraction order finder uses the NetworkX \cite{SciPyProceedings_11} implementation of the Louvain community detection algorithm, and an implementation of the QuickBB algorithm by Stephan Oepen \cite{mtool}. We ran QuickBB greedily on the communities and the metagraph assuming their small size in terms of nodes and edges would ensure the first tree decomposition found to already be a satisfying approximation. This allowed our order finder to return a contraction order and its associated contraction cost in a few seconds (for depth 20 Sycamore circuits). Thus to better the quality of the results, we run many order finding trials in parallel with different communities partitioning and return only the cheapest order.

For our pre-processing optimization strategy, we implemented a simulated annealing algorithm. We considered two graph-like diagrams to be neighbors if they differed by a single pivot. For about 100 steps starting with the original closed graph-like diagram we take one of its neighbor uniformly at random and choose, based on a global temperature and cost functions, whether or not to make it the current starting point for the next iteration (see Algorithm \ref{algo2}). We tried different ``proxy'' cost functions to guide the simulated annealing. We used treewidth-based heuristics such as min fill-in heuristics, QuickBB, mock contraction methods like $quick\_tw$, as well as an estimate for the contraction cost in FLOPs. We found out that pairing $quick\_tw$ and QuickBB to estimate $tw(G^\ast)$ yields the best results.

We find that this naive approach for finding equivalent tensor networks allows for significantly better contraction costs than existing state of the art methods.

\section{Benchmarks and Results}
\label{res}
In 2019, Arute \textit{et al.} \cite{arute2019quantum} sampled circuits on their 53-qubit Sycamore quantum chip in 200 seconds asserting quantum supremacy, as they initially estimated this task would require Summit, the world’s most powerful supercomputer today, approximately 10,000 years. Huang \textit{et al.} \cite{huang2020classical} since lowered the sampling task classical runtime to about 20 days using their tensor-network-based simulation engine AC-QDP.

In order to compare our approach to Huang et al. methods \cite{huang2020classical}, we generated 5 random Sycamore circuits for depths $\{12,14,16,18,20\}$ and performed the sampling task with AC-QDP. 
We then converted the same circuits into graph-like diagrams with and without applying our pre-processing before running AC-QDP. 
We averaged the results of the 5 circuits from each method, for each depth.
Figure \ref{tw} demonstrates the efficiency of our treewidth reduction strategy on the line-graph. The \textit{Standard} plot shows the value of $tw(G^{\ast})$ for the initial tensor network that would be used by AC-QDP, the \textit{ZX Unoptimized} plot shows our method without using further rewriting rules and local complementation strategies, and the \textit{ZX Optimized} plot shows the final results of our pre-processing.

\begin{figure}[!htb]
\centering
\includegraphics[scale=0.7]{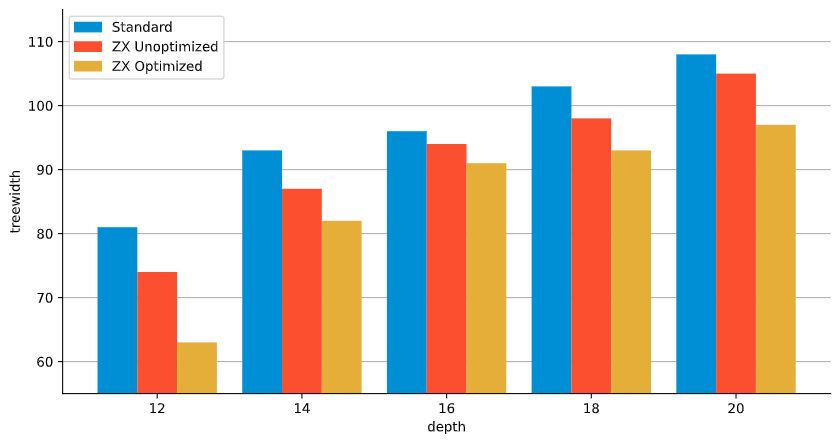}
\caption{Evolution of $tw(G^\ast)$ (min fill-in heuristic).}
\label{tw}
\end{figure}

Our pre-processing and order finding methods prove to be significantly better than the current state of the art method before slicing. Figure \ref{real} shows the number of FLOPs required for contracting the tensor network associated to the random Sycamore circuit. The AC-QDP data plotted here is taken directly from Huang et al. \cite{huang2020classical}. The \textit{simulated annealing} plot shows the result of our simplification strategy, the \textit{greedy algorithm} plot shows the behavior of the simulated annealing algorithm when the temperature is 0. The \textit{peak performance} plot displays our best results among all the circuits tested for each depth.

\begin{figure}[h!]
\centering
\includegraphics[scale=0.45]{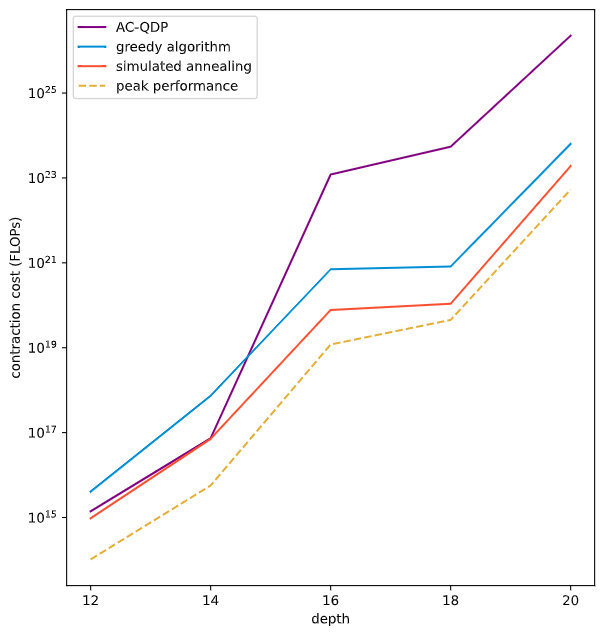}
\caption{Contraction cost before slicing. \textit{AC-QDP} data extrapolated from their ``Contraction cost and FLOPS efficiency for one sample" multiplied by their ``\textbf{\#} Sub-tasks".}
\label{real}
\end{figure}

\begin{table}[h!]
\begin{center}
\begin{tabular}{|c || c c c c c|} 
 \hline
 Depth & 12 & 14 & 16 & 18 & 20 \\ [0.5ex] 
 \hline\hline
 AC-QDP & $1.40\times10^{15}$ & $7.33\times10^{16}$ & $1.21\times10^{23}$ & $5.43\times10^{23}$ & $2.23\times10^{26}$\\ [0.5ex] 
 \hline
 Our method (average) & $9.61\times10^{14}$ & $7.15\times10^{16}$ & $7.75\times10^{19}$ & $1.09\times10^{20}$ & $1.89\times10^{23}$\\ [0.5ex]
 \hline
 Gain & $1.45 \times$ & $1.03 \times$ & $1561.3\times$ & $5122.6\times$ & $1179.9\times$\\ [0.5ex] 
\hline
Our method (peak) & $1.03\times10^{14}$ & $5.68\times10^{15}$ & $1.19\times10^{19}$ & $4.52\times10^{19}$ & $5.32\times10^{22}$\\ [0.5ex]
 \hline
 Gain & $13.5\times$ & $12.9\times$ & $10168.1\times$ & $12013.3\times$ & $4200.6\times$\\ [0.5ex] 
\hline
\end{tabular}
\caption{Average and peak contraction cost or number of FLOPs required for the contraction of random Sycamore circuits before the slicing step, Gain = AC-QDP / Our method.}
\end{center}
\end{table}
\vspace{1cm}

Running our order finder in parallel to improve efficiency, it yields competitive results in minutes : on the largest depth 20 Sycamore circuits, we find that our method reduces the number of FLOPs by a factor of about 1180 on average compared the AC-QDP's method, possibly turning a \textit{year-scale} computation into a \textit{day-scale} computation.

\section{Discussion}
\label{disc}
We now present the results of the actual contraction of our optimized tensor networks against the current state of the art method. All the computations were performed using the same hardware and the same random Sycamore circuits. The plot showing AC-QDP's performance has been computed using our own GPUs for fair comparison.

To measure the efficiency of our method, we ran our ZX-based optimization algorithms on the random Sycamore circuits as a pre-processing step before running AC-QDP on our hardware (NVIDIA Tesla V100S). We could not replicate their depth 20 Sycamore simulation result because the task we sampled did not terminate in 24 hours.

\begin{figure}[!htb]
\hspace{-25mm}
\includegraphics[scale=0.75]{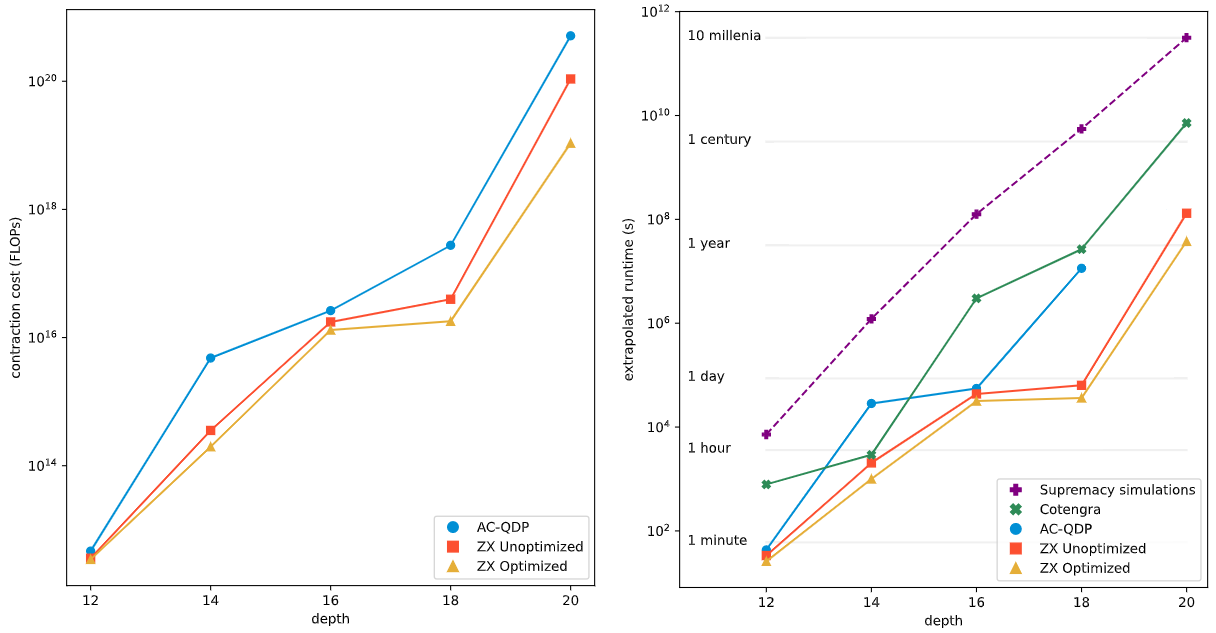}
\caption{Contraction cost and extrapolated runtime for the simulation of random Sycamore circuits. \textit{Supremacy simulations} data from \cite{arute2020supplementary}, \textit{Cotengra} data from \cite{gray2021hyper}, \textit{AC-QDP, ZX Unoptimized, ZX Optimized} simulated on our hardware.}
\end{figure}

\begin{table}[!h]
\begin{center}
\hspace*{-3mm}\begin{tabular}{|c || c c c c c|} 
 \hline
 Depth & 12 & 14 & 16 & 18 & 20 \\ [0.5ex] 
 \hline\hline
 Contraction Cost & $3.448\times10^{12}$ & $1.976\times10^{14}$ & $1.309\times10^{16}$ & $1.791\times10^{16}$ & $1.079\times10^{19}$\\ [0.5ex] 
 \hline
 \# Subtasks & $2^{5}$ & $2^{10}$ & $2^{16}$ & $2^{16}$ & $2^{23}$\\[0.5ex] 
 \hline
 Extrapolated Runtime & $25.98s$ & $998.71s$ & $8.79h$ & $10.09h$ & $435.72d$\\ [0.5ex] 
\hline
\end{tabular}
\caption{Extrapolated benchmarks for the simulation of random circuit sampling on Sycamore using our pre-processing method before running AC-QDP.}
\end{center}
\end{table}

In comparison to previous methods, our tensor networks require far less slicing: for the largest circuits, we used up to $32$ times less sub-tasks than AC-QDP, after only minutes of order finding on the graph-like ZX-diagram.

This work illustrates how formal methods can be used to simplify tensor networks. The results obtained are encouraging, and can be further improved. In particular, better ``proxy'' cost functions can be found for guiding the pivot-based heuristic and a more efficient (possibly parallelized) tree decomposition heuristic would directly lead to better contraction orders via Markov and Shi's method. On the theoretical side, more work has to be done to define a bound on the treewidth of the line-graph of all graph-like diagrams accessible by local complementation. Indeed, we could imagine refining the bound derived by Markov and Shi by considering the minimum of $tw(G'^\ast)$ on the set of graphs $G'$ that are equivalent to $G$ by local complementations. Even though finding this global minimum might be hard, it would help further tighten the complexity bound on the contraction cost of a tensor network.


\paragraph{Acknowledgements.} We would like to thank our colleagues at Atos Quantum Lab for support and helpful discussions, Vivien Vandaele for his help in ZX-Calculus, Océane Koska and Maxime Oliva for early reviews of this work. This work is part of HQI initiative (\url{www.hqi.fr}) and is supported by France 2c030 under the French National Research Agency award number "ANR-22-PNCQ-0002".

\printbibliography

\newpage
\appendix
\section{Additionnal definitions}
\label{appendix}

\begin{algorithm}
\caption*{Functions used in simulated annealing for graph-like optimization}
\begin{algorithmic}
\Function{Temperature}{$\mathrm{prog}$} \Comment{Temperature function for simulated annealing}
  \State { $\triangleright$ $\mathrm{prog} = step / nb\_steps$, progress of the simulated annealing}
  \State { $\triangleright$ Returns progress on decreasing exponential normalized from 1 to 0}
  \State \Return{$\dfrac{e^{\text{-}\mathrm{prog}} - \nicefrac{1}{e}}{1 - \nicefrac{1}{e}}$}
\EndFunction
\Statex
\Function{Accept}{cost, new\_cost, $\tau$} \Comment{Probability to accept new solution}
\State { $\triangleright$ cost, new\_cost = energy$_{current}$, energy$_{candidate}$}
  \If{new\_cost $<$ cost}
    \State \Return{$1$} \Comment{Greedily accepts better solutions}
  \Else
    \State \Return{$e^{\text{-}log(log(\text{new\_cost}) - log(\text{cost}) + 1) / \tau}$}
  \EndIf
\EndFunction
\end{algorithmic}
\end{algorithm}

\begin{figure}[h!]
\centering
\includegraphics[scale=0.85]{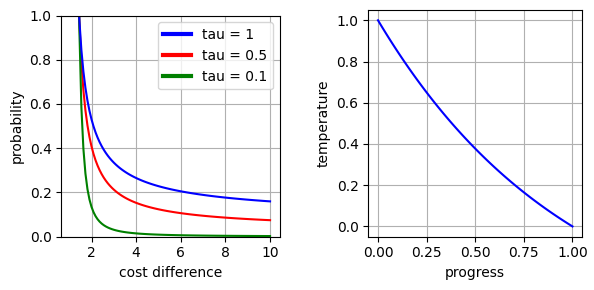}
\caption{\textsc{Accept} and \textsc{Temperature} functions for simulated annealing}
\end{figure}

\begin{algorithm}
\caption*{Implementation of graph-like unfusion based rules}
\begin{algorithmic}
\Function{MatchingUnfusion}{$G$, $node$} \Comment{Uses cycle basis algorithm to find the best couples of  nodes to preserve when unfusioning a high degree node}
  \State {$cycles_{node} = $ set of cycles in \Call{CycleBasis}{$G$} containing $node$}
  \State {$H = $ new graph}
  \State {$N = $ set of neighbors of $node$}
  \For{$u$ in $N$}
    \For{$v \neq u$ in $N$}
        \State {$in\_common = $ set of cycles in $cycles_{node}$ containing both $u$ and $v$}
        \State {Add weighted edge $(u,v)$ to $H$ of weight length($in\_common$)}
    \EndFor
  \EndFor
  \State \Return{\Call{MaxWeightMatching}{$H$}}
\EndFunction
\Statex
\Function{Unfusion}{$graph-like$, $node$} \Comment{Turns a high degree node (greater than 4) into a ternary-tree-like structure, preserves graph-likeness}
  \State {$G = $ underlying graph from input $graph-like$}
  \State {$matching = $ \Call{MatchingUnfusion}{$G$,$node$}}
  \While{$matching$ is not empty}
    \State {$edge (u,v) = matching$.pop()}
    \State {Remove edge $(node,u)$ from $G$}
    \State {Remove edge $(node,v)$ from $G$}
    \Statex
    \State {Add node with fresh name $w$ to $G$} \Comment{The unfusioned node}
    \State {Add edge $(u,w)$ to $G$}
    \State {Add edge $(v,w)$ to $G$}
    \Statex
    \State {Add node with fresh name $m$ to $G$} \Comment{Added to preserve graph-likeness}
    \State {Add edge $(w,m)$ to $G$}
    \State {Add edge $(m,node)$ to $G$}
  \EndWhile
\EndFunction
\Statex
\Function{Splitting}{graph-like} \Comment{Decreases the degree of the graph-like to 3, preserves graph-likeness}
  \State {$G = $ underlying graph from input $graph-like$}
  \State {$high\_degree = $ set of nodes of $G$ of degree greater than 3}
  \While{$high\_degree$ is not empty}
    \State {$node = high\_degree$.pop()}
    \State {\Call{Unfusion}{graph-like,$node$}}
    \State {Recompute $high\_degree$}
  \EndWhile
\EndFunction
\end{algorithmic}
\end{algorithm}

\begin{algorithm}
\caption*{Precontraction method and quick\_tw heuristic}
\begin{algorithmic}
\Function{Precontract}{$network$}
  \State { $\triangleright$ Compresses the network into a smaller, much denser one that encapsulates the structure of the original network and preserves the treewidth of its line-graph}
  \State {$G = $ underlying graph from input $network$}
  \Statex
  \State { $\triangleright$ First merging all leaves to their neighbor recursively}
  \State {$leaves = $ set of all degree 1 nodes in $G$}
  \While{$leaves$ is not empty}
    \State {$leaf = leaves$.pop()}
    \State {$parent = $ parent of $leaf$}
    \State {Contract $leaf$ to $parent$ in $G$}
    \If{degree of $parent$ in $G$ is 1}
      \State {Add $parent$ to $leaves$}
    \EndIf
  \EndWhile
  \Statex
  \State { $\triangleright$ Then contracting non-triangle edges}
  \State {$candidates = $ set of all edges such that both ends are degree 2 nodes}
  \While{$candidates$ is not empty}
    \State {$edge = candidates$.pop()} \Comment{$edge = (u,v)$}
    \State {Contract $v$ to $u$ in $G$}
    \For{\textbf{each} neighbor $w$ of $u$}
      \If{edge $(v,w)$ is in $candidates$} \Comment{Update the edges name}
        \State {Remove edge $(v,w)$ from $candidates$}
        \State {Add edge $(u,w)$ to $candidates$}
      \EndIf
    \EndFor
  \EndWhile
\EndFunction
\Statex
\Function{quick\_tw}{$network$} \Comment{Cost function used for simulated annealing}
  \State {$G = $ underlying graph from input $network$}
  \State {\Call{Precontract}{$G$}}
  \State {$lg = $ line-graph of $G$}
  \State \Return{\Call{QuickBB}{$lg$}}
\EndFunction
\end{algorithmic}
\end{algorithm}

\end{document}